\documentclass{amsart}
\usepackage{amsmath,amssymb}
\usepackage{xcolor}
\def\Z{\mathbb{Z}}

\def\R{\mathbb{R}}

\def\F{\mathcal{F}}
\newcommand{\C}{\mathcal{C}}
\newcommand{\rel}{\mathcal{R}}

\newcommand{\N}{\mathbb{N}}
\newcommand{\p}{\mathcal{P}}

\newcounter{thm}
\newcounter{ex}
\newcounter{re}
\newtheorem{Theorem}[thm]{Theorem}
\newtheorem{Lemma}[thm]{Lemma}

\newtheorem{Proposition}[thm]{Proposition}

\newtheorem{remark}[thm]{Remark}
\newtheorem{Definition}[thm]{Definition}

\title[The KP hierarchy on non-formal operators]{On the Cauchy problem for a Kadomtsev-Petviashvili hierarchy on non-formal operators and its relation with a group of diffeomorphisms}
\author[J.-P. Magnot and E.G. Reyes]{Jean-Pierre Magnot$^1$ and 
Enrique G. Reyes$^2$}

\address{\small $^1$: {LAREMA, Universit\'e d’Angers, 2 Bd Lavoisier, 
49045 Angers cedex 1, France and Lyc\'ee Jeanne d'Arc, 40 avenue de Grande Bretagne, 63000 Clermont-Ferrand, France}}
\email{\small magnot@math.univ-angers.fr; jean-pierr.magnot@ac-clermont.fr}
\address{\small $^2$:
	Departamento de Matem\'{a}tica y Ciencia de la Computaci\'{o}n,
	Universidad de Santiago de Chile, Casilla 307 Correo 2, Santiago,
	Chile. }\email{\small enrique.reyes@usach.cl;
	e\_g\_reyes@yahoo.ca}
\begin{document}

\begin{abstract}
{
We establish a rigorous link between infinite-dimensional regular 
Fr\"olicher Lie groups built out of non-formal pseudodifferential 
operators and the Kadomtsev-Petviashvili hierarchy.  
We introduce a version of the Kadomtsev-Petviashvili hierarchy on a regular Fr\"olicher Lie group of series of non-formal odd-class 
pseudodifferential operators.  We 
solve its corresponding Cauchy problem, and we establish a link between the dressing operator for our hierarchy and the action of diffeomorphisms and non-formal Sato-like operators on jet spaces. In appendix, we describe the group of Fourier integral operators in which this correspondence seems to take place. Also, motivated by Mulase's work on the KP hierarchy, we prove a group factorization theorem for
our group of Fourier integral operators. 
}
\end{abstract}

\maketitle


\smallskip

\smallskip


\section{Introduction}
The  Kadomtsev-Petviashvili hierarchy (KP hierarchy, for short)
is a system of nonlinear differential equations on an infinite number of dependent variables, each of which depend on
infinitely many independent variables. It reads as follows:
\begin{equation} \label{kpintro}
\frac{\partial}{\partial t_n} L = [(L^n)_+ , L] = (L^n)_+ \cdot L - L \cdot (L^n)_+ \; ,
\end{equation} 
in which
\[
L = \frac{\partial}{\partial x} + u_1 \left( \frac{\partial}{\partial x} \right)^{-1} + u_2 \left( \frac{\partial}{\partial x} \right)^{-2}
+ \cdots 
\; ,
\]
$u_1$, $u_2$, $\cdots$ are dependent variables, $(L^n)_+$
indicates the projection of the product $L^n = L \dots L$ on the space of differential operators, and $t_1$, $t_2$, $\cdots$ denote independent variables. An object such as $L$ above is a {\em formal pseudodifferential operator}. It is known that the set of formal pseudodifferential operators can be equipped with an associative algebra structure, see \cite{D}, and therefore (\ref{kpintro})
makes sense, at least, in an algebraic context. The reader is referred to \cite[Chp. 1, 5]{D} for a thorough algebraic discussion of KP and other important hierarchies.

The KP hierarchy  
is related to several soliton equations: for example, it contains the Korteweg-de Vries hierarchy and more generally the Gelfand-Dickey hierarchies, see \cite{D}. Moreover, 
it is {\em universal}. In Mulase's words, ``the KP system is the master equation for the largest possible family of
iso-spectral deformations of arbitrary ordinary differential operators'', see \cite[Section 3]{M2} and \cite{M3}; see also \cite[Corollary 6.2.8]{D} for another expression of this universality.
Solutions to KP can be recovered from quantum field theory and algebraic geometry among other fields, see for instance \cite{Kaz,Mick,M2} and references therein, and it can be posed for instance in contact geometry, see \cite{MR}. 

Can we solve Equation (\ref{kpintro}), in the sense of understanding its associated Cauchy problem? Yes. In the 1980's 
 Mulase published several fundamental papers on the
algebraic structure and formal integrability properties of the KP hierarchy, see \cite{M1,M2,M3}. A common
theme in these papers was the use of a powerful algebraic theorem on the factorization of a group of formal pseudodifferential operators
of infinite order which integrates the algebra of formal
pseudodifferential operators: this factorization ---a delicate
algebraic generalization of the Birkhoff decomposition of loop
groups appearing for example in \cite{PS}--- allowed him to
{\em solve} the Cauchy problem for the KP hierarchy in an
algebraic setting. A review of this theorem is in \cite{ER2013}.  
Mulase's results have been re-interpreted and extended in the context of 
(generalized) differential geometry on diffeological and Fr\"olicher 
spaces, and they have been used to prove well-posedness of the KP 
hierarchy in analytic categories, see \cite{ERMR,Ma2013,MR2016} and our 
recent review \cite{MR2019}. 

It is important to point out that in the above mentioned papers the 
operators under consideration are {\em formal} pseudodifferential 
operators: they are not understood as operators 
acting on smooth maps or smooth sections of vector bundles. 
They differ from non-formal pseudodifferential operators 
by (unknown) smooth kernel operators, the so-called smoothing operators. 
As is well-known, any classical non-formal 
pseudodifferential operator $A$ generates a formal operator (the one 
obtained from the asymptotic expansion of the 
symbol of $A$, see \cite{ARS2,ARS3,Gil}), but there is no canonical way 
to recover a non-formal operator from a formal one.

Can we introduce and discuss a version of the KP hierarchy
using classical non-formal  pseudo-differential operators? Yes. 
{\em The aim of this paper is to show that Equation 
$(\ref{kpintro})$ can indeed be posed and solved on regular 
Fr\"olicher Lie groups built with the help of a particular class of non-
formal pseudo-differential operators. } 
Our first motivation for considering this problem comes from the 
following observation: 
pushing forward equations onto a quotient of a relation of equivalence 
is easy and unambiguous (up to compatibility conditions), while 
pulling-back equations from a quotient space to full space can be often 
performed in very many ways. As explained in the previous paragraph, the 
KP hierarchy can  be understood as posed on a quotient space of 
classical pseudodifferential operators, and so it would be very natural 
to aim at proposing a version of KP using the pseudodifferential 
operators themselves.  
Our second motivation for considering non-formal 
pseudodifferential operators comes from our previous work 
\cite{MR2016}. In this reference we use versions of ``dressing 
operators'' for equation (\ref{kpintro}), and we obtain solutions 
to KP with the help of an operator which acts on initial conditions 
(see \cite[Section 4]{MR2016}). It is natural to wonder if it is 
possible to understand these operators in a non-formal setting. 

{
What class of pseudodifferential operators can we use, in 
order to write down an equation such as (\ref{kpintro})? {\em We 
work with  odd-class classical pseudodifferential operators which act 
on smooth sections of a given trivial (finite rank) vector bundle 
$S^1 \times V$}.  
These pseudodifferential operators were first considered by Kontsevich 
and Vishik in \cite{KV1,KV2} in order to deal with spectral functions 
and renormalized determinants. We use them in two ways:
\begin{itemize}
\item We take them as building blocks for our non-formal KP hierarchy. One reason why odd-class pseudodifferential operators are natural to use
in this context is the fact that differential operators are all
odd-class, and so we can indeed hope to pose Equation (\ref{kpintro}) with their help.
\item We build a central extension of $Diff_+(S^1)$ by a group of bounded odd-class classical pseudodifferential operators, in which $Diff_+(S^1)$ is the group 
of orientation-preserving diffeomorphisms of $S^1$. We present this construction because the structure
of this central extension allows us to prove rather easily a
Mulase-type factorization theorem in our non-formal context, an
observation we think is interesting of its own\footnote{Elements of this central extension are Fourier
integral operators called 
$Diff_+(S^1)-$pseudodifferential operators. To the best
of our knowledge, groups of $Diff(S^1)-$pseudodifferential operators were independently described (with $S^1$ replaced by 
a compact Riemannian manifold $M$) in \cite{Ma2016}, in the context of differential geometry of non-parametrized, 
non-linear grassmannians, and in \cite{Pay2014} as a possible structure group on which Chern-Weil constructions could be performed.}.  
\end{itemize}

\smallskip

We organize our work as follows. Section 1 is this introduction.
Section 2 is a short review on Fr\"olicher Lie groups, mostly inspired by \cite{MR2016,MR2019}.    
In this paper we consider several infinite-dimensional groups built with the help of non-formal pseudodifferential operators. Some of these groups are beyond the reach of traditional analytic means but they do possess Fr\"olicher structures, and so it is natural to begin with a review of the Fr\"olicher setting. Section 3 is on 
Fr\"olicher Lie groups of Fourier integral and pseudodifferential operators, following mainly \cite{Ma2016}. References for the analytic tools used therein are \cite{BGV,Gil,Scott}. 
Then, in Section 4 we propose our version of KP hierarchy: we consider the Lie algebra  $Cl_{h,odd}(S^1,V)$ of formal power series in a parameter $h$ whose coefficients are classical odd-class pseudodifferential
operators satisfying some technical conditions. These conditions allow us to find a {\em regular} (a notion explained in Section 2)
Fr\"olicher Lie group which integrates the Lie algebra $Cl_{h,odd}(S^1,V)$. In this extended context we can pose and solve the Cauchy problem for KP. In Section 4 we also highlight a non-formal operator $U_h \in Cl_{h,odd}(S^1,V)$ which depends on the initial
condition of our KP hierarchy; this operator generates its solutions very much in the spirit of the standard theory of 
$R$-matrices, see \cite{ER2013,MR2016} and references therein.   
Then in Section \ref{S1} we show how to 
recover the operator $U_h$ by analysing the Taylor expansion of functions in the image of the twisted operator
$A : f \in C^\infty(S^1; V ) \mapsto S_0^{-1}(f) \circ g$, in which $g \in Diff_+(S^1)$ and $S_0$ is our version of a
``dressing operator" as considered for example in \cite[Chapter 6]{D}. 
Finally, we include an appendix in which we introduce a group of Fourier
integral operators, the central extension of  
$Diff_+(S^1)$ by the group $Cl^{0,*}_{odd}(S^1,V)$ of all 
odd-class, invertible and bounded, classical pseudodifferential operators. As mentioned above, working with this central extension we can prove a non-formal analogue of the Mulase decomposition of
\cite{M1,M2,M3}.
}

\section{Preliminaries on categories of regular Fr\"olicher Lie groups} \label{regular}
In this section we recall briefly the formal setting which allows us to work rigorously
with (Lie) groups of pseudodifferential operators. No new statements are given here: we follow the expositions appearing in \cite{Ma2013,MR2019,MR2016}.
 We begin with the notion of a diffeological space:

\begin{Definition} Let $X$ be a set.
	
	\noindent $\bullet$ A \textbf{p-parametrization} of dimension $p$ 
	on $X$ is a map from an open subset $O$ of $\R^{p}$ to $X$.
	
	\noindent $\bullet$ A \textbf{diffeology} on $X$ is a set $\p$
	of parametrizations on $X$ such that:
	
	- For each $p\in\N$, any constant map $\R^{p}\rightarrow X$ is in $\p$;
	
	- For each arbitrary set of indexes $I$ and family $\{f_{i}:O_{i}\rightarrow X\}_{i\in I}$
	  of compatible maps that extend to a map $f:\bigcup_{i\in I}O_{i}\rightarrow X$,
	  if $\{f_{i}:O_{i}\rightarrow X\}_{i\in I}\subset\p$, then $f\in\p$.
	
	- For each $f\in\p$, $f : O\subset\R^{p} \rightarrow X$, and $g : O' \subset \R^{q} \rightarrow O$, in which $g$ is  
	a smooth map (in the usual sense) from an open set $O' \subset \R^{q}$ to $O$, we have $f\circ g\in\p$.
	
	\vskip 6pt If $\p$ is a diffeology on $X$, then $(X,\p)$ is
	called a \textbf{diffeological space} and, if $(X,\p)$ and $(X',\p')$ are two diffeological spaces, 
	a map $f:X\rightarrow X'$ is \textbf{smooth} if and only if $f\circ\p\subset\p'$. 
\end{Definition} 

The notion of a diffeological space is due to J.M. Souriau, see \cite{Sou}; see also \cite{Chen} 
for related constructions, and \cite{Igdiff} for a contemporary point of view. Of particular interest to us is the following subcategory of the category of diffeological
spaces. 

\begin{Definition} 
A \textbf{Fr\"olicher} space is a triple $(X,\F,\C)$ such that
	
	- $\C$ is a set of paths $\R\rightarrow X$,
	
	- $\F$ is the set of functions from $X$ to $\R$, such that a function
	$f:X\rightarrow\R$ is in $\F$ if and only if for any
	$c\in\C$, $f\circ c\in C^{\infty}(\R,\R)$;
	
	- A path $c:\R\rightarrow X$ is in $\C$ (i.e. is a \textbf{contour})
	if and only if for any $f\in\F$, $f\circ c\in C^{\infty}(\R,\R)$.
	
	\vskip 5pt If $(X,\F,\C)$ and $(X',\F',\C ')$ are two
	Fr\"olicher spaces, a map $f:X\rightarrow X'$ is \textbf{smooth}
	if and only if $\F'\circ f\circ\C\subset C^{\infty}(\R,\R)$.
\end{Definition}

This definition first appeared in \cite{FK}; we use terminology 
borrowed from Kriegl and Michor's book \cite{KM}.
A short comparison of the notions of diffeological and Fr\"olicher spaces is in \cite{Ma2006-3}; the reader can also 
see \cite{Ma2013,Ma2018-2,MR2016,Wa} for extended expositions.
In particular, it is explained in \cite{MR2016} that 
{\em Fr\"olicher and Gateaux smoothness are the same notion if we 
restrict to a Fr\'echet context.}

Any family of maps $\F_{g}$ from $X$ to $\R$ generates a 
Fr\"olicher structure $(X,\F,\C)$ by setting, after \cite{KM}:

- $\C=\{c:\R\rightarrow X\hbox{ such that }\F_{g}\circ c\subset C^{\infty}(\R,\R)\}$

- $\F=\{f:X\rightarrow\R\hbox{ such that }f\circ\C\subset C^{\infty}(\R,\R)\}.$

We call $\F_g$ a \textbf{generating set of functions}
for the Fr\"olicher structure $(X,\F,\C)$. One easily see that
$\F_{g}\subset\F$. 
A Fr\"olicher space $(X,\F,\C)$
carries a natural topology, the pull-back topology of
$\R$ via $\F$. In the case of a finite dimensional
differentiable manifold $X$ we can take $\F$ as the set of all smooth
maps from $X$ to $\R$, and $\C$ the set of all smooth paths from
$\R$ to $X.$ Then, the underlying topology of the
Fr\"olicher structure is the same as the manifold topology
\cite{KM}. 

We also remark that if $(X,\F, \C)$ is a Fr\"olicher space, we can
define a natural diffeology on $X$ by using the following family
of maps $f$ defined on open domains $D(f)$ of Euclidean spaces, see \cite{Ma2006-3}:
$$
\p_\infty(\F)=
\coprod_{p\in\N}\{\, f: D(f) \rightarrow X; \, \F \circ f \in C^\infty(D(f),\R) \quad \hbox{(in
	the usual sense)}\}.$$

If $X$ is a finite-dimensional differentiable manifold, this diffeology is
called the { \em n\'ebuleuse diffeology}, see \cite{Sou}. Now,
we can easily show the following:

\begin{Proposition} \label{fd} \cite{Ma2006-3}
	Let$(X,\F,\C)$
	and $(X',\F',\C')$ be two Fr\"olicher spaces. A map $f:X\rightarrow X'$
	is smooth in the sense of Fr\"olicher if and only if it is smooth for
	the underlying diffeologies $\p_\infty(\F)$ and $\p_\infty(\F').$
\end{Proposition}

Thus, Proposition \ref{fd} and the foregoing remarks imply that 
the following implications hold:
\vskip 12pt

\begin{tabular}{ccccc}
	smooth manifold  & $\Rightarrow$  & Fr\"olicher space  & $\Rightarrow$  & diffeological space
\end{tabular}
{
\vskip 12pt These implications can be refined. The reader is
referred to the Ph.D. thesis \cite{Wa} for a deeper analysis of 
them. 

\begin{remark}
The set of contours $\C$ of the Fr\"olicher space $(X,\F,\C)$ does 
not give us a diffeology, because a diffeology needs to be stable 
under restriction of domains. In the case of paths in $\C$ the 
domain is always $\R$ whereas the domain of 1-plots can (and has 
to) be any interval of $\R.$ However, $\C$ defines a ``minimal 
diffeology'' $\p_1(\F)$ whose plots are smooth parametrizations 
which are locally of the type $c \circ g,$ in which 
$g \in \p_\infty(\R)$  and $c \in \C.$ Within this setting, we can 
replace $\p_\infty$ by $\p_1$ in Proposition $\ref{fd}$. The main
technical tool needed to discuss this issue is Boman's theorem 
\cite[p.26]{KM}.
Related discussions are in \cite{Ma2006-3,Wa}.
\end{remark}

\begin{Proposition} \label{prod1} Let $(X,\p)$ and $(X',\p')$
	be two diffeological spaces. There exists a diffeology 
	$\p\times\p'$ on
	$X\times X'$  made of plots $g:O\rightarrow X\times X'$
	that decompose as $g=f\times f'$, where $f:O\rightarrow X\in\p$
	and $f':O\rightarrow X'\in\p'$. We call it the \textbf{product diffeology}, 
	and  this construction extends to an infinite product.
\end{Proposition}

We apply this result to the case of Fr\"olicher spaces and we 
derive (compare with e.g. \cite{KM}) the following:

\begin{Proposition} \label{prod2} 
Let $(X,\F,\C)$ and $(X',\F',\C')$ be two Fr\"olicher spaces 
equipped with their natural diffeologies $\p$ and $\p'$ . There is 
a natural structure of Fr\"olicher space on $X\times X'$ which 
contours $\C\times\C'$ are the 1-plots of $\p\times\p'$. 
\end{Proposition}

We can also state the above result for infinite products;
we simply take Cartesian products of the plots, or of the contours.

\smallskip

Now we consider quotients after \cite{Sou} and 
\cite[p. 27]{Igdiff}: Let $(X,\p)$ be 
a diffeological space, and let $X'$ be a set. Let 
$f:X\rightarrow X'$ be a map. We define the \textbf{push-forward 
diffeology} as the coarsest (i.e. the smallest for inclusion) among 
the diffologies on $X'$, which contains $f \circ \p.$ 

\begin{Proposition} \label{quotient} 
Let $(X,\p)$ b a diffeological space and $\rel$ an equivalence 
relation on $X$. Then, there is a natural diffeology on $X/\rel$, 
noted by $\p/\rel$, defined as the push-forward diffeology on 
$X/\rel$ by the quotient projection $X\rightarrow X/\rel$. 
\end{Proposition}

Given a subset $X_{0}\subset X$, where $X$ is a Fr\"olicher space
or a diffeological space, we equip $X_{0}$ with structures 
induced by $X$ as follows:
\begin{enumerate}
\item If $X$ is equipped with a diffeology $\p$, we define
a diffeology $\p_{0}$ on $X_{0}$ called the \textbf{subset or trace 
	diffeology}, see \cite{Sou,Igdiff}, by setting 
\[ 
\p_{0}=
\lbrace p\in\p \hbox{ such that the image of }p\hbox{ is a subset 
of } X_{0}\rbrace\; .
\]
\item If $(X,\F,\C)$ is a Fr\"olicher space, we take as a 
generating set of maps $\F_{g}$ on $X_{0}$ the restrictions of the 
maps $f\in\F$. In this case, the contours (resp. the induced 
diffeology) on $X_{0}$ are the contours (resp. the plots) on $X$ 
whose images are a subset of $X_{0}$.
\end{enumerate}

\smallskip

Our last general construction is the so-called functional 
diffeology. 
{
	Its existence implies the following crucial fact: the category of 
diffeological spaces is Cartesian 
closed, something which is certainly not true in the category of
smooth manifolds.} Our discussion follows \cite{Igdiff}. 
Let $(X,\p)$ and $(X',\p')$ be diffeological spaces. 
Let $M \subset C^\infty(X,X')$ be a set of smooth maps. 
The \textbf{functional diffeology} on $S$ is the diffeology $\p_S$
made of plots 
$$ \rho : D(\rho) \subset \R^k \rightarrow S$$
such that, for each $p \in \p$, the maps 
$\Phi_{\rho, p}: (x,y) \in D(p)\times D(\rho) \mapsto 
\rho(y)(x) \in X'$ are plots of $\p'.$ 
We have, see \cite[Paragraph 1.60]{Igdiff}:

\begin{Proposition} \label{cvar} 
Let $X,Y,Z$ be diffeological spaces. Then, 
$$
C^\infty(X\times Y,Z) = C^\infty(X,C^\infty(Y,Z)) = 
C^\infty(Y,C^\infty(X,Z))
$$
as diffeological spaces equipped with functional diffeologies.
\end{Proposition}

\smallskip
\noindent
{
Now, given an algebraic structure, we can define a
corresponding compatible diffeological (Fr\"olicher) structure, see 
for instance \cite{Les}. For example, see \cite[pp. 66-68]{Igdiff},
if $\R$ is equipped with its canonical diffeology (Fr\"olicher
structure), we say that an $\R-$vector space equipped with a 
diffeology (Fr\"olicher structure) is called a
diffeological (Fr\"olicher) vector space if addition and scalar 
multiplication are smooth. We state:
}

\begin{Definition}
Let $G$ be a group equipped with a diffeology (Fr\"olicher 
structure). We call it a \textbf{diffeological (Fr\"olicher) group} 
if both multiplication and inversion are smooth.
\end{Definition}

\noindent
Since we are interested in infinite-dimensional analogues of Lie 
groups, we need to consider tangent spaces of 
diffeological spaces, and we have to deal with Lie algebras and 
exponential maps. We state, after \cite{DN2007-1} and \cite{CW}the 
following definition:

\begin{Definition}
\begin{enumerate}
\item[(i)] For each $x\in X,$ we consider 
$$
C_{x}=\{c \in C^\infty(\R,X)| c(0) = x\}
$$ 
and we take the equivalence relation $\mathcal{R}$ given by 
$$
c\mathcal{R}c' \Leftrightarrow \forall f \in C^\infty(X,\R), \partial_t(f \circ c)|_{t = 0} = \partial_t(f \circ c')|_{t = 0}.
$$ 
Equivalence classes of $\mathcal{R}$ are called {\bf germs} and are 
denoted by $\partial_t c(0)$ or $\partial_tc(t)|_{t=0}$.  
The {\bf internal tangent cone} at $x$ is the quotient 
$^iT_xX = C_x / \mathcal{R}.$ If 
$X = \partial_tc(t)|_{t=0} \in {}^iT_X, $ we define the derivation 
$Df(X) = \partial_t(f \circ c)|_{t = 0}\, .$
\item[(ii)] The \textbf{internal tangent space} at $x \in X$ is the 
vector space generated by the internal tangent cone.
\end{enumerate}
\end{Definition}

The reader may compare this definition to the one appearing in 
\cite{KM} for manifolds in the ``convenient"  $c^\infty-$setting. 
The internal tangent cone at a point $x$ is not a vector space in 
many examples; 
this motivates item (ii) above, see \cite{CW,DN2007-1}. 
Fortunately, the internal tangent cone at $x\in X$ {\em is} a vector 
space for the objects under consideration in this work, see Proposition 
\ref{leslie} below; it will be called, simply, the tangent space at $x 
\in X$.

Following Iglesias-Zemmour, see \cite{Igdiff}, we do not assert that 
arbitrary diffeological groups have associated Lie algebras; however, 
the following holds, see \cite[Proposition 1.6.]{Les} and 
\cite[Proposition 2.20]{MR2016}.

\begin{Proposition} \label{leslie}
Let $G$ be a diffeological group. Then the tangent cone at the neutral 
element $T_eG$ is a diffeological vector space.
\end{Proposition}

\noindent
{
The proof of Proposition \ref{leslie}
appearing in \cite{MR2016} uses explicitly the diffeologies
$\mathcal{P}_1$ and  $\mathcal{P}_\infty$ which appear in 
Proposition 3 and Remark 4 of this work.
}

\begin{Definition}
The diffeological group $G$ is a \textbf{diffeological Lie group} 
if and only if 
the Adjoint action of $G$ on the diffeological vector space 
$^iT_eG$ defines a Lie bracket. 
In this case, we call $^iT_eG$ the Lie algebra of $G$ and we denote 
it by $\mathfrak{g}.$
\end{Definition}

Let us concentrate on Fr\"olicher Lie groups, following  
\cite{Ma2013} and \cite{Les}. If $G$ is a Fr\"olicher Lie group 
then, after (i) and (ii) above we have that: 
	$$
\mathfrak{g} = \{ \partial_t c(0) ; c \in \C \hbox{ and } c(0)=e_G \}
	$$
	is the space of germs of paths at $e_G.$ Moreover:
	\begin{itemize}
		\item Let $(X,Y) \in \mathfrak{g}^2,$ $X+Y = \partial_t(c.d)(0)$  
		where $c,d \in \C ^2,$ $c(0) = d(0) =e_G ,$
		$X = \partial_t c(0)$ and $Y = \partial_t d(0).$
		\item Let $(X,g) \in \mathfrak{g}\times G,$ $Ad_g(X) = \partial_t(g c g^{-1})(0)$  where $c \in \C ,$ $c(0) =e_G ,$
		and $X = \partial_t c(0).$
		\item Let $(X,Y) \in \mathfrak{g}^2,$ $[X,Y] = \partial_t( Ad_{c(t)}Y)$   where $c \in \C ,$ $c(0) =e_G ,$
		$X = \partial_t c(0).$
	\end{itemize}
	All these operations are smooth and thus well-defined as operations on Fr\"olicher spaces, see 
	\cite{Les,Ma2013,Ma2018-2,MR2016}.

The basic properties of adjoint, coadjoint actions, and of Lie brackets, remain globally the same
as in the case of finite-dimensional Lie groups, and the proofs are similar: see \cite{Les} 
and \cite{DN2007-1} for details. 

\begin{Definition} \label{reg1} \cite{Les} 
A Fr\"olicher Lie group $G$ with Lie algebra $\mathfrak{g}$
	is called \textbf{regular} if and only if there is a smooth map 
	\[
Exp:C^{\infty}([0;1],\mathfrak{g})\rightarrow C^{\infty}([0,1],G)
	\]
	such that $g(t)=Exp(v(t))$ is the unique solution
	of the differential equation \begin{equation}
	\label{loga}
	\left\{ \begin{array}{l}
	g(0)=e\\
	\frac{dg(t)}{dt}g(t)^{-1}=v(t)\end{array}\right.\end{equation}
	We define the exponential function as follows:
	\begin{eqnarray*}
		exp:\mathfrak{g} & \rightarrow & G\\
		v & \mapsto & exp(v)=g(1) \; ,
	\end{eqnarray*}
	where $g$ is the image by $Exp$ of the constant path $v.$ 
\end{Definition}

When the Lie group $G$ is a vector space $V$, the notion of regular Lie group specialize to what is called 
{\em regular vector space} in \cite{Ma2013} and {\em integral vector space} in \cite{Les}; we follow the latter 
terminology.

\begin{Definition} \label{reg2} \cite{Les}
	Let $(V,\p)$ be a Fr\"olicher vector space. 
The space $(V,\p)$ is \textbf{integral} if there is a smooth map
$$ 
\int_0^{(.)} : C^\infty([0;1];V) \rightarrow C^\infty([0;1],V)
$$ 
such that $\int_0^{(.)}v = u$ if and only if $u$ is the unique solution of the differential equation
	\[
	\left\{ \begin{array}{l}
	u(0)=0\\
	u'(t)=v(t)\end{array}\right. .\]
\end{Definition}

This definition applies, for instance, if $V$ is a complete locally convex topological vector space equipped with its 
natural Fr\"olicher structure given by the Fr\"olicher completion of its n\'ebuleuse diffeology, see
\cite{Igdiff,Ma2006-3,Ma2013}.

\begin{Definition}
Let $G$ be a Fr\"olicher Lie group with Lie algebra $\mathfrak{g}$. 
Then, $G$ is regular with integral Lie algebra if  
$\mathfrak{g}$ is integral and $G$ is regular in the sense of 
Definitions $\ref{reg1}$ and $\ref{reg2}$.
\end{Definition}

We finish this section with two structural results proven in 
\cite{Ma2013}. The first one provides us with an example of a 
Fr\"olicher Lie group (instances of which appear prominently in the 
analysis of the Cauchy problem for the Kadomtsev-Petviashvili 
carried out in \cite{Ma2013,MR2016}),
while the second one is used in the construction of regular Lie 
groups of non-formal pseudodifferential and Fourier operators, see
\cite{Ma2013,Ma2016} and Section 3 below.

\begin{Theorem} \label{regulardeformation} 
Let $(A_n)_{n \in \mathbb{N}^*} $ be a sequence of integral 
(Fr\"olicher) vector spaces equipped with a graded smooth 
multiplication operation on $\bigoplus_{n \in \mathbb{N}^*} A_n ,$ 
i.e. a multiplication such that for each $n,m \in \mathbb{N}^*$, 
$A_n .A_m \subset A_{n+m}$ is smooth with respect to the 
corresponding Fr\"olicher structures.  
Let us define the (non unital) algebra of formal series:
$$
\mathcal{A}= \left\{ \sum_{n \in \mathbb{N}^*} a_n | 
\forall n \in \mathbb{N}^* , a_n \in A_n \right\}\; ,
$$
equipped with the Fr\"olicher structure of the infinite product.
Then, the space 
	$$1 + \mathcal{A} = 
\left\{ 1 + \sum_{n \in \mathbb{N}^*} a_n | 
\forall n \in \mathbb{N}^* , a_n \in A_n \right\} $$
	is a regular Fr\"olicher Lie group with 
  	integral Fr\"olicher Lie algebra $\mathcal{A}.$
Moreover, the exponential map defines a smooth bijection 
$\mathcal{A} \rightarrow 1+\mathcal{A}.$
\end{Theorem}

\begin{Theorem}\label{exactsequence} 
Let
$$ 
1 \longrightarrow K \stackrel{i}{\longrightarrow} G \stackrel{p}{\longrightarrow}  H \longrightarrow 1 
$$
be an exact sequence of Fr\"olicher Lie groups, such that there is 
a smooth section $s : H \rightarrow G,$ and such that the trace 
diffeology  from $G$ on $i(K)$ coincides with the push-forward 
diffeology from $K$ to $i(K).$ We consider also the corresponding 
sequence of Lie algebras
$$ 
0 \longrightarrow \mathfrak{k} \stackrel{i'}{\longrightarrow} 
\mathfrak{g} \stackrel{p}{\longrightarrow}  
	\mathfrak{h} \longrightarrow 0 \; . 
$$
Then,
\begin{itemize}
\item The Lie algebras $\mathfrak{k}$ and $\mathfrak{h}$ are 
integral if and only if the Lie algebra $\mathfrak{g}$ is integral;
\item The Fr\"olicher Lie groups $K$ and $H$ are regular if and 
only if the Fr\"olicher Lie group $G$ is regular.
\end{itemize}
\end{Theorem}

A result similar to Theorem \ref{exactsequence} is also 
valid for Fr\'echet Lie groups, see \cite{KM}.
   
\section{Preliminaries on pseudodifferential operators} \label{PDO}

We introduce the groups and algebras
of non-formal pseudodifferential operators needed to set up our
version of the KP hierarchy.  Basic definitions are 
valid for a real or complex finite-dimensional vector bundle $E$ 
over $S^1$; below (see paragraph ``Notations'') 
we specialize our considerations to the case $E = S^1 \times V$ 
in which $V$ is a finite-dimensional {complex} vector space. 
The following definition appears in \cite[Section 2.1]{BGV}.

\begin{Definition} 
The graded algebra of differential operators acting on the space of 
smooth sections $C^\infty(S^1,E)$ is the algebra $DO(E)$ generated 
by:
	
$\bullet$ Elements of $End(E),$ the group of smooth maps $E \rightarrow 
E$ leaving each fibre globally invariant and which restrict to linear 
maps on each fibre. This group acts on sections of $E$ via (matrix) 
multiplication;
	
	$\bullet$ The differentiation operators
$$\nabla_X : g \in C^\infty(S^1,E) \mapsto \nabla_X g$$ where $\nabla$ 
	is a connection on $E$ and $X$ is a vector field on $S^1$.
\end{Definition}

Multiplication operators are operators of order $0$; differentiation 
operators and vector fields are operators of order 1. In local 
coordinates, a differential operator of order $k$ has the form
$ P(u)(x) = \sum p_{i_1 \cdots i_r} \nabla_{x_{i_1}} \cdots 
\nabla_{x_{i_r}} u(x) \; , \quad r \leq k \; ,$
in which the coefficients $p_{i_1 \cdots i_r}$ can be matrix-valued.
We note by $DO^k(S^1)$,$k \geq 0$, the differential operators of order 
less or equal than $k$.

The algebra $DO(E)$ is graded by order. It is a subalgebra of the algebra of classical pseudodifferential operators
$Cl(S^1,E),$ an algebra that contains, for example, the square root of the Laplacian, its inverse, and all
trace-class operators on $L^2(S^1,E).$  Basic facts on pseudodifferential operators defined 
on a vector bundle $E \rightarrow S^1$ can be found for instance 
in \cite{Gil} and in the review \cite{Pay2014}. 
A global symbolic calculus for pseudodifferential operators has been defined independently 
by J. Bokobza-Haggiag, see \cite{BK} and H. Widom, see \cite{Wid}. In these papers is shown how the 
geometry of the base manifold $M$ furnishes an obstruction to generalizing 
local formulas of composition and inversion of symbols; we do not recall these
formulas here since they are not involved in our computations. 

Following \cite[Section 1]{Ma2006-2}, see also \cite{Ma2016}, we 
assume henceforth that $S^1$ is equipped with charts such that the 
changes of coordinates are translations. We also restrict our 
considerations to complex vector bundles over $S^1$. It is 
well-known that they are trivial, i.e. $E = S^1 \times V.$ Taking 
this fact into account, we use the following notational 
conventions:

\vskip 6pt
\noindent
\textbf{Notations.} 
We note by  $ PDO (S^1, V) $ 
(resp.  $ PDO^o (S^1, V)$, resp. $Cl(S^1,V)$) the space of 
pseudodifferential operators (resp. pseudodifferential operators of 
order $o$, resp. classical pseudodifferential operators) acting on 
smooth sections of $E$, and by 
$Cl^o(S^1,V)= PDO^o(S^1,V) \cap Cl(S^1,V)$ the space of classical 
pseudodifferential operators of order $o$. We also denote by 
$Cl^{o,\ast}(S^1,V)$ the group of units of $Cl^o(S^1,V)$.

\vskip 6pt

A topology on spaces of classical pseudo differential operators has 
been described by Kontsevich and Vishik in \cite{KV1}; see also 
\cite{CDMP,PayBook,Scott} for other descriptions.
We use all along this work the Kontsevich-Vishik topology. This is 
a Fr\'echet topology such that each space 
$Cl^o(S^1,V)$ is closed in $Cl(S^1,V).$ 
We set
$$ PDO^{-\infty}(S^1,V) = \bigcap_{o \in \Z} PDO^o(S^1,V) \; .$$
It is well-known that $PDO^{-\infty}(S^1,V)$ is a two-sided ideal 
of $PDO(S^1,V)$, see e.g.
\cite{Gil,Scott}. 
Therefore, we can define the quotients
$$\mathcal{F}PDO(S^1,V) = PDO(S^1,V) / PDO^{-\infty}(S^1,V),$$ 
$$\F Cl(S^1,V) = Cl(S^1,V) /  PDO^{-\infty}(S^1,V),$$
$$ \quad \F Cl^o(S^1,V) = Cl^o(S^1,V)  / PDO^{-\infty}(S^1,V)\; .$$
The script font $\F$ stands for {\it  formal pseudodifferential 
operators}. The quotient $\mathcal{F}PDO(S^1,V)$ is an algebra 
isomorphic to the set of formal symbols, see \cite{BK}, 
and the identification is a morphism of $\mathbb{C}$-algebras for 
the usual multiplication on formal symbols (see e.g. \cite{Gil}). 

A known result on the structure of the spaces we are using is the
following. 

\begin{Theorem} 
The groups $Diff_+(S^1)$, $Cl^{0,*}(S^1,V),$  and 
$\mathcal{F}Cl^{0,*}(S^1,V)$, in which  
${\mathcal F}Cl^{0,*}(S^1,V)$ is the group of units of the algebra 
${\mathcal F}Cl^{0}(S^1,V)$, are regular Fr\'echet Lie groups.
\end{Theorem}

Indeed, it follows from \cite{Ee,Om} that $Diff_+(S^1)$ is open in 
the Fr\'echet manifold $C^\infty(S^1,S^1)$. This fact makes it a 
Fr\'echet manifold and, following \cite{Om}, a regular Fr\'echet 
Lie group. The same result follows from the discussion appearing in 
\cite[Section III.3]{Neeb2007}. Also, it is noticed in 
\cite{Ma2006} that the results of \cite{Gl2002} imply that the 
group $Cl^{0,*}(S^1,V)$ (resp. $\F Cl^{0,*}(S^1,V)\,$) is open in 
$Cl^0(S^1,V)$ (resp. $\F Cl^{0}(S^1,V)\,$) and that it is a regular 
Fr\'echet Lie group. This fact is also discussed in 
\cite[Proposition 4]{Pay2014}. 
Our comments after Definition 2, see also Remark 2.6 in 
\cite{MR2016}, imply that these groups are also regular Fr\"olicher 
Lie groups.

\begin{Definition} \label{d7} 
A classical pseudodifferential operator $A$ on $S^1$ is called odd 
class if and only if for all $n \in \Z$ and all 
$(x,\xi) \in T^*S^1$ we have:
	$$ \sigma_n(A) (x,-\xi) = (-1)^n  \sigma_n(A) (x,\xi)\; ,$$
	in which $\sigma_n$ is the symbol of $A$. 
\end{Definition}

This {particular} class of pseudodifferential operators has been 
introduced in \cite{KV1,KV2}; it is also called the ``even-even 
class'', see \cite{Scott}. We will follow the terminology 
of the first two references: hereafter, the notation $Cl_{odd}$ 
will refer to odd class classical pseudodifferential operators. 

\smallskip

We will need the 
following result, intrinsically present in \cite{KV1,Scott} which
we prove quickly:

\begin{Lemma}
$Cl_{odd}(S^1,V)$ and $Cl^0_{odd}(S^1,V)$ are associative algebras.
\end{Lemma}
\begin{proof}
That $Cl_{odd}(S^1,V)$ is an associative algebra can be checked thus: we 
consider the map $\phi$ given by
\begin{equation}\label{phi} 
\begin{array}{cccl} \phi: & Cl(S^1,V) &\rightarrow&  \mathcal{F}Cl(S^1,V) \\ & A &\mapsto & 
\sum_{n \in \mathbb{N}}\sigma_{-n}(x,\xi) - (-1)^n\sigma_{-n}(x,-\xi)\; . \end{array} 
\end{equation}
Then, $$Cl_{odd}(S^1,V)=Ker(\phi)\; ,$$
which proves the first claim. 
That $Cl^0_{odd}(S^1,V)$ is an associative algebra follows from 
the previous result and the standard fact that zero-order classical 
pseudodifferential  operators form an algebra, see for
instance \cite{Pay2014}.	
\end{proof}

{
Now we observe that because of the symmetry property stated in 
Definition \ref{d7}, an odd class pseudodifferential operator $A$ has a 
partial symbol of non-negative order $n$ that reads
\begin{equation} \label{alfa}
\sigma_{n}(A)(x,\xi) = \gamma_n(x) (i\xi)^n \, ,
\end{equation} 
where $\gamma_n \in C^\infty(S^1,L(V))$.  This consequence of 
Definition $\ref{d7}$ allows us to check the following result, which 
is of importance for the upcoming description of our KP hierarchy: 

\begin{Proposition}\label{SD}
	The space of odd class classical pseudodifferential operators 
	satisfies the direct sum decomposition 
\begin{equation} \label{beta0}
Cl_{odd}(S^1,V) = Cl_{odd}^{-1}(S^1,V) \oplus DO(S^1,V)\; .
\end{equation}
\end{Proposition}
}

\smallskip

We finish this section with a proposition which singles out an 
interesting Lie group included in $Cl_{odd}(S^1,V)$.

\begin{Proposition}
	The algebra $Cl_{odd}^0(S^1,V)$ is a closed subalgebra of 
	$Cl^0(S^1,V)$. Moreover, $Cl_{odd}^{0,*}(S^1,V)$ is 
	\begin{itemize}
		\item an open subset of $Cl^0_{odd}(S^1,V)$ and,
		\item a regular Fr\'echet Lie group.
	\end{itemize} 
\end{Proposition}
\begin{proof}
	We note by $\sigma(A)(x,\xi)$ the total formal symbol of 
	$A \in Cl^0(S^1,V).$ Similarly as in (\ref{phi}) we let 
$$
\phi: Cl^0(S^1,V)\rightarrow \mathcal{F}Cl^{0}(S^1,V)
$$ defined by 
$$\phi(A) = 
\sum_{n \in \mathbb{N}}\sigma_{-n}(x,\xi) - (-1)^n\sigma_{-n}(x,-\xi).
$$
This map is smooth, and $$Cl^{0}_{odd}(S^1,V)= Ker(\phi),$$ which 
shows that 	$Cl_{odd}^0(S^1,V)$ is a closed subalgebra of   
$Cl^0(S^1,V).$ 
		Moreover, if $H = L^2(S^1,V),$ 
		$$Cl^{0,*}_{odd}(S^1,V) = Cl^{0}_{odd}(S^1,V)\cap GL(H),$$
		which proves that $Cl_{odd}^{0,*}(S^1,V)$ is open in the 
		Fr\'echet algebra $Cl^0(S^1,V),$ and it follows 
		that it is a regular Fr\'echet Lie group by arguing along 
		the lines of \cite{Gl2002,Neeb2007}.
	\end{proof}


\section{The h-KP hierarchy with non-formal odd-class operators} 
\label{KP}

First of all let us make some comments on the spaces just introduced. 
In order to find an analogue to Equation (\ref{kpintro}) we need to 
consider a space of pseudodifferential operators which is close 
with respect to taking powers of operators.
Since the space of odd-class pseudodifferential operators 
is an associative algebra, we can take this class as the space in which 
the dependent variable appearing in Equation (\ref{kpintro}) lives. 
Proposition \ref{SD} implies 
that we have the diagram of Lie groups and Lie algebras

$$
\begin{array}{ccccc}
Cl^{-1,\ast}_{odd}(S^1,V) & \rightarrow & Cl^{\ast}_{odd}(S^1,V) & \rightarrow  & (H \quad ?) \\
\downarrow &  & \downarrow &  &  \\
\mathbb{C} Id \oplus Cl^{-1}_{odd}(S^1,V) & \rightarrow & Cl_{odd}(S^1,V) &
\rightarrow & DO(S^1,V) \\
 & & \| & & \\
 & & Cl^{-1}_{odd}(S^1,V) \oplus DO(S^1,V) & & 
\end{array}
$$

\smallskip

{
The problem is to find a suitable Fr\"olicher Lie group $H$.
If it were possible, we could set up an equation of the form 
$$
\frac{\partial}{\partial t_n} L = [(L^n)_D , L] 
$$ 
for fixed $n$, where $(.)_D$ denotes projection into the space of 
differential operators, and try to study its corresponding Cauchy problem with the help of a factorization theorem, as in our previous
papers \cite{Ma2013,MR2019,MR2016}. Now, in these articles we find a 
regular Fr\"olicher Lie group $H$ with Lie algebra the space of
differential operators 
by using formal differential operators of infinite order but, if we
proceed in this way in the present context, we would leave the
framework of non-formal pseudodifferential operators. Thus, instead of 
doing this we use series, motivated by \cite{Ma2013,MR2019} 
and \cite[Subsection 4.2]{MR2016}.
}

\begin{Definition} \label{dfs}
Let $h$ be a formal parameter. The set of odd class 
$h$-pseudodifferential operators is the set of formal series 
\begin{equation} \label{alpha}
Cl_{h,odd}(S^1,V)=
\left\{ \sum_{n \in \mathbb{N}} a_n h^n \, | \, 
a_n \in Cl_{odd}^{n}(S^1,V)  \right\}\; .
\end{equation}
\end{Definition}

We have the following result on the structure of 
$Cl_{h,odd}(S^1,V)$:

\begin{Theorem} \label{fs}
The set $Cl_{h,odd}(S^1,V)$ is a Fr\'echet algebra, and its group 
of units given by 
\begin{equation} \label{beta}	
	Cl_{h,odd}^*(S^1,V)=
	\left\{ \sum_{n \in \mathbb{N}} a_n h^n \, | \, a_n \in 
Cl_{odd}^{n}(S^1,V), a_0\in Cl_{odd}^{0,*}(S^1,V)  \right\} \; ,
\end{equation}
is a regular Fr\'echet Lie group.
\end{Theorem}
\begin{proof}
As we showed in Proposition 20 (and it follows from the work 
\cite{Gl2002} by Gl\"ockner) the group $Cl^{0,*}_{odd}(S^1,V)$ is a 
regular Fr\'echet Lie group since it is open in 
$Cl^0_{odd}(S^1,V)$. According to classical properties of 
composition of pseudodifferential 
operators \cite{Scott}, see also \cite{KV1}, the natural 
multiplication on $Cl^{0,*}_{odd}(S^1,V)$ is smooth for the 
	product topology inherited from the classical topology on 
	classical pseudodifferential operators, and inversion is smooth 
	using the classical formulas of inversion of series. 
	In this way we conclude that $Cl_{h,odd}(S^1,V)$ is a Fr\'echet 
	algebra. 
	
Moreover, the series 
$\sum_{n \in \mathbb{N}} a_nh^n \in Cl_{h,odd}(S^1,V)$ is 
invertible if and only if 
$a_0 \in Cl_{odd}^{0,*}(S^1,V),$ which shows that 
$Cl_{h,odd}^*(S^1,V)$ is open in $Cl_{h,odd}(S^1,V)$. The same 
result quoted before, from \cite{Gl2002}, ends the proof.
\end{proof}

\begin{remark}
The assumption $a_n \in Cl^n_{odd}$ in Definition $\ref{dfs}$ and 
Theorem $\ref{fs}$ can be relaxed to the condition 
$$a_0 \in Cl^{0,*}_{odd} \hbox{ and } \forall n \in \N^*, 
a_n \in Cl_{odd}\; ;$$ 
this is sufficient for having a regular Lie group. 
The more stringent growth conditions imposed 
in {\rm (\ref{alpha})} and {\rm (\ref{beta})} will ensure 
regularity {\em and} they will allow us to use arguments borrowed
from \cite[Subsection 4.1]{MR2016} for proving existence and smoothness 
of solutions to our KP hierarchy to be introduced next. 
\end{remark}

Now we need to split the algebra $Cl_{h,odd}(S^1,V)$. We do
so in a very straightforward way: since an operator 
$A \in Cl_{odd}(S^1,V)$ splits into $A = A_S + A_D$, in which 
$A_S \in Cl^{-1}_{odd}(S^1,V)$ and $A_D \in DO(S^1,V)$, see Proposition 26, we have, for 
$A = \sum_{n \in \mathbb{N}} a_nh^n \in Cl_{h,odd}(S^1,V),$ 
 the decomposition $A = A_S + A_D$ with 
 $$A_S = \sum_{n \in \mathbb{N}} (a_n)_Sh^n$$
 and  $$A_D = \sum_{n \in \mathbb{N}} (a_n)_Dh^n\; .$$
We set $DO_h(S^1,V) = \{  \sum_{n \in \mathbb{N}} a_n h^n :
a_n \in DO(S^1,V) \}$.

\smallskip

We now introduce our version of the KP hierarchy with non-formal 
pseudodifferential operators. 
Let us assume that $t_1, t_2, \cdots, t_n , \cdots,$ are an 
infinite number of different formal variables which 
will become the independent variables of our equation. We make the 
following definition:

\begin{Definition} \label{33}
Let $S_0 \in Cl^{-1,*}_{odd}(S^1,V)$ and let 
$L_0 = S_0 (h\frac{d}{dx})S_0^{-1}.$
We say that an operator 
$$ 
L(t_1,t_2,\cdots) \in  Cl_{h,odd}(S^1,V)[[ht_1,...,h^nt_n...]]
$$ satisfies the $h-$deformed KP hierarchy if and only if 
\begin{equation} \label{jph}
	\left\{\begin{array}{cl} 
	L(0) = & L_0 \\
	\frac{d}{dt_n}L =& \left[(L^n)_D, L\right] \; .
	\end{array}
	\right.
\end{equation}
\end{Definition}

Let us make some comments on Definition \ref{33}. We have 
followed Mulase, see \cite{M1,M3}, in fixing the ``time dependence'' of 
the dependent variable via series. Thus, our equation (\ref{jph}) is 
written for a dependent variable of the form
$$
L(t_1,t_2,\cdots) = \sum_{s \in T} L_s \tau^s \; ,
$$
in which $L_t \in Cl_{h,odd}(S^1,V)$, 
$\tau = (ht_1) (h^2 t_2)\cdots$ and for $s$ in $T$,
$s = (\alpha_1,\alpha_2, \cdots)$  (in which $\alpha_i \in \N$, 
$\alpha_i \neq 0$ just for finite number of indexes $i$) we define
$\tau^s = (ht_1)^{\alpha_1} (h^2 t_2)^{\alpha_2}\cdots$.
This series can be understood as a smooth function on the 
algebraic sum 
\begin{equation} \label{times}
T= \bigoplus_{n \in \mathbb{N}^*}(\mathbb{R}t_n)
\end{equation}
for the product topology and product Fr\"olicher structure, 
see Proposition \ref{prod2} and \cite{Ma2013,MR2016}. The ``space  
dependence'', on the other hand, is fixed with the help of a 
derivation on $S^1$ which in standard coordinates (see Section 3)
reads $d/dx$. Finally, we stress the fact that we are 
scaling our variables via
$$\left\{ \begin{array}{ccc}t_n& \mapsto & h^nt_n \\
\frac{d}{dx} & \mapsto &h\frac{d}{dx} \end{array}\right. \; .$$
Our reason to do this is that we need to work with regular 
Fr\"olicher Lie groups, and this scaling allows us to do so, as we 
explain in \cite{Ma2013,MR2016}.

\smallskip

{
In this context, we have a ``Mulase factorization", in the spirit of
\cite{M1,M3} and \cite{Ma2013,MR2016}, } which 
looks schematically as follows:

$$
\begin{array}{ccccc}
Cl^{-1,\ast}_{h,odd}(S^1,V) & \rightarrow & Cl^{\ast}_{h,odd}(S^1,V) 
& \rightarrow  & DO_h^\ast (S^1,V) \\
\downarrow &  & \downarrow &  & \downarrow \\
\mathbb{C} Id \oplus Cl^{-1}_{h,odd}(S^1,V) & \rightarrow 
& Cl_{h,odd}(S^1,V) &
\rightarrow & DO_h(S^1,V) 
\end{array}
$$
in which 
$$
DO_h^\ast(S^1,V) = \left\{ \sum_{n \in \N} a_n h^n | a_n \in 
DO(S^1,V)\,, a_0 \in DO^{0,\ast}(S^1,V) \right\}  \; .
$$

Now we solve the initial value problem (\ref{jph}). Since 
$S_0 \in CL_{odd}^{-1,\ast}(S^1,V)$ and $h (d/dx) \in DO_h(S^1,V)$,
we have that the initial condition $L_0$ belongs to 
$Cl_{h,odd}(S^1,V)$. Also, we need to use the operator
$U_h = \exp\left(\sum_{n \in N^*} h^nt_n (L_0)^n\right)$. We note that 
$\sum_{n \in N^*} h^n t_n (L_0)^n$ belongs to $Cl_{h,odd}(S^1,V)[[ht_1,h^2 
t_2, \cdots]]$. In the theorem below we consider this sum as a 
series and also as a smooth function with domain $T$ and image in 
$Cl_{h,odd}(S^1,V)$, in which $T$ is given by (\ref{times}).

\begin{Theorem} \label{hKP}
Let $U_h(t_1,...,t_n,...) = \exp\left(\sum_{n \in N^*} h^nt_n (L_0)^n\right) \in Cl_{h,odd}^\ast(S^1,V).$ 
	Then:
	\begin{itemize}
		\item There exists a unique pair $(S,Y)$ such that
		\begin{enumerate}
			\item $U_h = S^{-1}Y,$
			\item $Y \in Cl_{h,odd}^*(S^1,V)_D$
			\item $S \in Cl_{h,odd}^*(S^1,V)$ and $S - 1 \in Cl_{h,odd}(S^1,V)_S.$ 
		\end{enumerate}
	Moreover, the map 
	$$(S_0,t_1,...,t_n,...)\in Cl^{0,*}_{odd}(S^1,V)\times T \mapsto (U_h,Y)\in (Cl_{h,odd}^*(S^1,V))^2$$ is smooth.
		\item The operator $L \in Cl_{h,odd}(S^1,V)[[ht_1,...,h^nt_n...]]$ given by $L = S L_0 S^{-1} = Y L_0 Y^{-1}$, 
		is the unique solution to the hierarchy of equations 
		{\begin{equation}		\label{formalKP} 
		\left\{\begin{array}{ccl}
		\frac{d}{dt_n}L &=& \left[(L^n)_D(t), L(t)\right] =  -\left[(L^n)_S(t), L(t)\right]\\
		L(0) & = & L_0 \\
		\end{array}\right. \; ,
		\end{equation}}
		in which the operators in this infinite system are 
		understood as formal operators. 
		\item The operator $L \in Cl_{h,odd}(S^1,V)[[ht_1,...,h^nt_n...]]$ given by $L = S L_0 S^{-1} = Y L_0 Y^{-1}$ 
		is the unique solution of the hierarchy of equations 
		\begin{equation}
		\label{trueKP} \left\{\begin{array}{ccl}
		\frac{d}{dt_n}L &=& \left[(L^n)_D(t), L(t)\right] =  -\left[(L^n)_S(t), L(t)\right]\\
		L(0) & = & L_0 \\
		\end{array}\right. 
		\end{equation}
		in which the operators in this infinite system are understood as odd class, non-formal operators. 
	\end{itemize}
\end{Theorem}	
\begin{proof}
First of all, we consider $U_h$. Since
$$
U_h(t_1,...,t_n,...) = 
\exp\left(\sum_{n \in N^*} h^nt_n (L_0)^n\right) \in 
Cl_{h,odd}^\ast(S^1,V)[[ht_1,...,h^nt_n...]] \; ,
$$
we can write
$$
U_h = \sum_{s \in T} A_s (h \tau)^s \; ,
$$
in which $h \tau = (h t_1 , h^2 t_2 , h^3 t_3 , \cdots)$ and
$A_s \in Cl^\ast_{h,odd}(S^1 , V)$. In turn, for each $s \in T$
we can set $A_s = \sum_{n \in \N} a_{s n} h^n$, where 
$a_{s n} \in Cl^{n}_{odd}(S^1,V)$, $n\geq 1$ and $a_{s 0} \in Cl^{0,\ast}_{odd}(S^1,V)$. Thus, we have
$$
U_h = \sum_{s \in T} \left( \sum_{n \in \N} a_{s n} h^n \right)
(h \tau)^s \; .
$$
Now we observe that, since $a_{s n} \in Cl^{n}_{odd}(S^1,V)$,
the total symbol of $a_{sn}$ can be written as
$$
\sigma(a_{s n}) = \sum_{-\infty < k \leq n} a_{sn k}\xi^k
$$
in which $a_{snk} : S^1 \rightarrow \R \otimes End(V)$.
(The pass from pseudodifferential operators to symbols is discussed
in detail in, e.g., \cite[Section 2]{ARS2} and \cite[p. 55]{ARS3}).
This means that we can write
\begin{eqnarray}
\sigma(U_h) & = & 
\sum_{s \in T} \left( \sum_{n \in \N} 
\left( \sum_{-\infty < k \leq n} a_{sn k}\xi^k \right) h^n \right)(h \tau)^s
\nonumber \\
 & = & \sum_{n \in \N} \left[ \sum_{-\infty < k \leq n} \left( \sum_{s \in T} a_{sn k} (h \tau)^s \right) \xi^k \right] h^n \; .
 \label{sym}
\end{eqnarray}
Equation (\ref{sym}) tells us that $\sigma(U_h)$ belongs to the
algebra $\Psi_h(R)$, in which $R$ is the algebra of power series in
$\tau$ whose coefficients belong to the differential algebra of
smooth functions $C^\infty(S^1)\otimes End(V)$. See Definition 4.3
in \cite{MR2016}. (Also, we can say that $\sigma(U_h) \in 
\widetilde{\mathcal{A}}$, where $\widetilde{\mathcal{A}}$ is 
defined in Section 5.4 of \cite{Ma2013}). Now we use that
$a_{s 0} \in Cl^{0,\ast}_{odd}(S^1,V)$ and that therefore its total 
symbol is of the form 
$$
\sigma(a_{s 0}) = \sum_{-\infty < k \leq 0} a_{s0k} \xi^k
= a_{s00} + \sum_{-\infty < k \leq -1} a_{s0k} \xi^k
$$
in which $a_{s00}$ is invertible. Let us set
$$
a(\tau)_{nk} = \sum_{s \in T} a_{sn k} (h \tau)^s \; .
$$
It follows that $\sigma(U_h)$ can be written as
\begin{eqnarray}
\sigma(U_h) & = & 
\sum_{-\infty < k \leq 0} a(\tau)_{0k} \xi^k + \sum_{n \geq 1} \left[ \sum_{-\infty < k \leq n} a(\tau)_{nk} \xi^k \right] h^n   \nonumber \\
 & = & a(\tau)_{00} + \sum_{-\infty < k \leq -1} a(\tau)_{0k}\xi^k
 + \sum_{n \geq 1} \left[ \sum_{-\infty < k \leq n} a(\tau)_{nk} \xi^k \right] h^n \; .
\end{eqnarray}
Since $a_{s00}$ is invertible, we conclude that $\sigma(U_h)$
belongs to $G\Psi_q(R)$ (see Definition 4.3 in \cite{MR2016}; we 
can also say that $\sigma(U_h) \in G_{\mathcal{A}}$ in the notation 
of \cite{Ma2013}). Now we use Equation (4.14) of \cite{MR2016}.
There exist unique $S_{symb} \in G_{R,h}$ and $Y_{symb} \in \mathcal{D}_q(R)$ such that
$$
\sigma(U_h) = S_{symb}^{-1} Y_{symb}  \; .
$$ 
Now, there exist non-formal odd class operators $Y$ and $S$ defined up to smoothing operators such that $S_{symb} = \sigma(S)$ and
$Y_{symb} = \sigma(Y)$, and so we can write
$$
\sigma(U_h) = \sigma(S)^{-1} \sigma(Y)  \; .
$$
The symbol $\sigma(Y)$ is a formal series in $h, t_1, \cdots t_n,\cdots$ of symbols of differential operators, which are 
	in one-to-one correspondence with a series of (non-formal)
    differential operators. Thus, the operator $Y$ is uniquely defined, not up to a smoothing operator; it depends 
    smoothly on $U_h$, and so does $S = Y U_h^{-1}.$ This ends the proof of the first point. 
	
	The second point on the $h-$deformed KP hierarchy is proven along the lines of \cite{Ma2013,MR2016}, since it corresponds 
	essentially to an existence result for symbols. 
	
Finally, we prove the third point: We have that $L= Y L_0 Y^{-1}$ 
is well-defined and, following classical computations which can be 
found in e.g. \cite{ER2013,MR2016}, we have:
	 \begin{enumerate} \item
		$L^k = Y L_0^{\;k} Y^{-1}$
		\item  $U_h\,L_0^{\;k} U_h^{-1}
		= L_0^{\;k}$ since $L_0$ commutes with $U_h = \exp (\sum_k h^kt_k\,L_0^{\; k}).$
	\end{enumerate}
	
	It follows that $L^k = Y L_0^{\;k} Y^{-1} = W W^{-1}
	Y L_0^{\;k} Y^{-1} W W^{-1} = W L_0^{\;k} W^{-1}$.
	
	We take $t_k$-derivative of $U$ for each $k \geq 1$. We get the equation
	\[
	\frac{d U_h}{dt^k} = - W^{-1}\frac{dW}{dt_k} W^{-1} Y + S^{-1} \frac{dY}{dt_k}
	\]
	and so, using $U_h = S^{-1}\,Y$, we obtain the decomposition
	\[
	W L_0^{\; k} W^{-1} = - \frac{dW}{dt_k} W^{-1}  + \frac{dY}{dt_k} Y^{-1} \; .
	\]
	Since $\frac{dW}{dt_k} W^{-1} \in Cl_{h,odd}(S^1,V)_S$ and $\frac{dY}{dt_k} Y^{-1} \in
	Cl_{h,odd}(S^1,V)_D$, we conclude that
	$$(L^k)_D = \frac{dY}{dt_k} Y^{-1} \; \; \mbox{ and } \; \; (L^k)_S = - \frac{dW}{dt_k} W^{-1}.$$ Now we take 
	$t_k$-derivative of $L$:
	\begin{eqnarray*}
		\frac{d L}{d t_k} & = & \frac{dY}{dt_k} L_0 Y^{-1} - Y L_0 Y^{-1} \frac{dY}{dt_k}
		Y^{-1} \\
		& = & \frac{dY}{dt_k} Y^{-1} Y L_0 Y^{-1} - Y L_0 Y^{-1}\frac{dY}{dt_k}
		Y^{-1} \\
	    & = &{	 (L^k)_D\, L - L\, (L^k)_D } \\
	    & = &{	 [ (L^k)_D , L ] \; . }
	\end{eqnarray*}
	We check the initial condition: We have $L(0) = Y(0)
	L_0 Y(0)^{-1}$, but $Y(0) = 1$ by the definition of $U_h.$
	
	Smoothness with respect to the variables $(S_0, t_1,...,t_n,...)$ is already proved by construction, 
	and we have established smoothness of the map $L_0 \mapsto Y$ at the beginning of the proof. Thus, the map
	$$L_0 \mapsto L(t)= Y(t) L_0 Y^{-1}(t)$$ is smooth. The corresponding equation 
$$\frac{d}{dt_k}L = -\left[(L^k)_S,L\right]$$ is obtained the same way. 
	
It remains to check that the announced solution is the unique solution 
to the non-formal hierarchy	(\ref{trueKP}). This is still true at the 
formal level, but two solutions which differ by a smoothing operator may 
appear at this step of the proof. Let $(L+K)(t_1,...)$ be another 
solution, in which $K$ is a smoothing operator depending on the 
variables $t_1,...$, and $L$ is the solution derived from $U_h.$ 
Then, for each $n \in \N^*$ we have 
	 	$$(L+K)^n_D = L^n_D\; ,$$ 
	 	which implies that 
	 	$K$ satisfies the {\em linear} equation 
	 	$$ \frac{d K}{dt_n} = [L^n_D,K]$$ 
	 	with initial conditions 
	 	$K|_{t=0} = 0.$  We can construct the unique solution $K$ by 
	 	induction on $n$, beginning with $n=1$. 
	 	Let  $g_n$ be such that 
$$
(g_n^{-1} dg_n)(t_n) = L^n_D(t_1,...t_{n-1},t_n,0,...) \; .
$$ 
Then we get that 
$$
K(t_1,...t_n,0....) = Ad_{g_n(t_n)}\left(K(t_1,t_{n-1},0...)\right) \; , 
$$
and hence, by induction, 
$$K(0)=0 \Rightarrow K(t_1,0...)=0 \Rightarrow \cdots \Rightarrow 
K(t_1,...t_n,0....) =0 \Rightarrow \cdots \; ,
$$
which implies that $K=0.$   
\end{proof}

\section{KP equations and $Diff_{+}(S^1)$} \label{S1}

Let $A_0 \in Cl_{odd}^{-1}(S^1,V)$, and set $S_0 = \exp(A_0)$. The operator $S_0 \in Cl_{odd}^{-1,\ast}(S^1,V)$ is our 
version of the dressing operator of standard KP theory, see for instance \cite[Chapter 6]{D}. We define the operator $L_0$
by
$$
f \mapsto L_0(f) = h \left( S_0 \circ \frac{d}{dx} \circ S_0^{-1} \right) (f)
$$
for $f \in C^\infty(S^1,V)$. We note that $L_0^k(f) = h^k S_0 \frac{d^k}{dx^k}(S_0^{-1}(f))$, a formula which we will use
presently. Our aim is to connect the operator
$$
U_h = \exp\left(\sum_{n \in \N*} h^n t_n L_0^n\right) \; ,
$$
which generates the solutions of the $h-$deformed KP hierarchy described in Theorem \ref{hKP}, with the  Taylor expansion
of functions in the image of the twisted operator 
$$
A : f \in C^\infty(S^1,V) \mapsto S_0^{-1}(f) \circ g \; ,
$$
in which $g \in Diff_+(S^1)$. {
	We remark that 
$A \in Cl^{-1,*}_{odd}(S^1,V)$ for each $g \in Diff_+(S^1)$; 
our decomposition theorem proven in the appendix (see Theorem \ref{SY}) will imply that it
is also smooth with respect to $g$}. 

For convenience, we identify $S^1$ with $[0;2\pi[\sim \mathbb{R}/2\pi\mathbb{Z},$ assuming implicitly that all the 
values under consideration are up to terms of the form $2 k \pi,$ for $k \in \mathbb{Z}.$ Set $c = S_0^{-1}(f)\circ g
\in C^\infty(S^1,V)$. We compute:
\begin{eqnarray*}
c(x_0+h) & = & \left( S^{-1}_0(f) \circ g \right)(x_0 + h) \\ 
& \sim_{x_0} & \left( S^{-1}_0(f) \circ g \right)(x_0) + \sum_{n \in \mathbb{N}^*}
\left[ \frac{h^n}{n !}\,  \frac{d^n}{dx^n} \left( S^{-1}_0(f) \circ g \right) \right](x_0) \\ 
& = & \left(  S^{-1}_0(f) \circ g \right)(x_0) + \\
&   & \sum_{n \in \mathbb{N}^*} 
\left[\frac{h^n}{n!}\sum_{k = 1}^n B_{n,k}(u_1(x_0),...,u_{n-k+1}(x_0))\frac{d^k}{dx^k}\left(S_0^{-1}(f)\circ g\right)(x_0)\right] \; ,
\end{eqnarray*}
in which we have used the classical Fa\'a de Bruno formula for the higher chain rule in terms of Bell's polynomials 
$B_{n,k}$, and $u_i(x_0) = g^{(i)}(x_0)$ for $i = 1, \cdots n-k+1$. We can rearrange the last sum and write
\begin{eqnarray*}
c(x_0+h) & \sim_{x_0} & \left( S^{-1}_0(f) \circ g \right)(x_0) + \\
&  & \sum_{k \in \mathbb{N}^*} \sum_{n \geq k} \left[ \frac{h^n}{n !}\, B_{n,k}(u_1(x_0),...,u_{n-k+1}(x_0)) \frac{d^k}{dx^k}\left(S_0^{-1}(f)\right)\right] (g(x_0))  
\end{eqnarray*}
or,
\begin{eqnarray} \label{t1}
c(x_0+h) & \sim_{x_0} & \sum_{k \in \mathbb{N}} \left[ a_k h^k \frac{d^k}{dx^k}\left(S_0^{-1}(f)\right)\right] (g(x_0))  
\end{eqnarray}
in which $a_0=1$ and 
$$ a_k = \sum_{n \geq k} \frac{h^{n-k}}{n!} B_{n,k}(u_1(x_0),...,u_{n-k+1}(x_0)) $$
for $k \geq 1$. In terms of the operator $L_0$, Equation (\ref{t1}) means that 
\begin{eqnarray} \label{t2}
c(x_0+h) & \sim_{x_0} & S_0^{-1} \sum_{k \in \mathbb{N}} \left[ a_k\, L_0^k(f)\right] (g(x_0))  \; .
\end{eqnarray}
We now define the sequence $(t_n)_{n \in \N^*}$ by the formula
\begin{equation} \label{t3} 
\log\left( \sum_{k \in \N} a_kX^k \right)  = \sum_{n \in \N^*} t_n X^n \; ,
\end{equation}
so that both, $a_k$ and $t_n$, are series in the variable $h$. We obtain
\begin{eqnarray*} 
c(x_0+h) & \sim_{x_0} & S_0^{-1} \exp \left( \sum_{n \in \mathbb{N}^*} \frac{t_n}{h^n}\, L_0^k(f) \right) (g(x_0))  \; .
\end{eqnarray*}
We state the following theorem:
\begin{Theorem}
	Let $f \in C^\infty(S^1,V)$ and set $c = S_0^{-1}(f) \circ g \in C^\infty(S^1,V).$ 
	The Taylor series at $x_0$ of the function $c$ is given by
	$$c(x_0+h) \sim_{x_0} S_0^{-1} \left( U_h(t_1/h,t_2/h^2,...)(f) \right)(g(x_0)) \; ,$$
	in which the times $t_i$ are related to the derivatives of $g$ via Equation $(\ref{t3})$.
\end{Theorem}
The coefficients of the series $a_k$ and $t_n$ appearing in (\ref{t3}) depend smoothly on $g \in Diff_+(S^1)$ and
$x_0 \in S^1$. Indeed, the map
$$
(x,g) \in S^1 \times Diff_+(S^1) \mapsto \left( g(x), (u_n(x))_{n \in \N^*} \right) \in S^1 \times \R^{\N^*}
$$
is smooth due to Proposition \ref{prod2} (more precisely, due to the generalization of Proposition \ref{prod2} to
infinite products); smoothness $a_k$ then follows, while smoothness of $t_n$ is consequence of
Equation (\ref{t3}).

\medskip

\begin{remark} 
As a by-product of the foregoing computations, we notice the following relation. If $f \in C^\infty(S^1,V),$ 
we can write 
$$ f(x_0+h) \sim_{x_0} f(x_0)+ \sum_{n \in \mathbb{N}^*}\left(\frac{h^n}{n!}\left(\frac{d}{dx}\right)^n f\right) (x_0) = 
   \left(\exp\left(h\frac{d}{dx}\right)f\right)(x_0) \in J^\infty(S^1,V)$$
for $x_0 \in S^1.$
Thus, the operator $\exp\left(h\frac{d}{dx}\right)$ belongs to the space $Cl_h(S^1,V).$
\end{remark}

\section*{Appendix:the group of $\boldsymbol{Diff_+(S^1)-}$pseudodifferential operators}
Now we present a restricted class of groups of 
Fourier integral operators which we will call 
$Diff_+(S^1)$-pseudodif\-fer\-en\-tial operators following \cite{Ma2016}.
These groups appear as central extensions of $Diff_+(S^1)$ by groups of (often bounded) pseudodifferential operators. 
We do not state the basic facts on Fourier integral operators here (they can be found in the classical paper \cite{Horm}),
but we recall the following theorem, which was stated in \cite{Ma2016} for a general base manifold $M$.

\begin{Theorem} \label{DiffPDO} \cite[Theorem 4]{Ma2016}
	Let $H$ be a regular Lie group of pseudodifferential operators acting on smooth sections of  a trivial bundle 
	$E \sim V \times S^1 \rightarrow S^1.$
	The group $Diff(S^1)$ acts smoothly on  $C^\infty(S^1,V),$ and it is assumed to act smoothly on $H$ by adjoint action.
	If $H$ is stable under the $Diff(S^1)-$adjoint action, then there exists a regular Lie group $G$ of Fourier integral
	operators defined through the exact sequence:
	$$ 
1 \rightarrow H \rightarrow G \rightarrow Diff(S^1) \rightarrow  1
	\; .
	$$ 
If $H$ is a Fr\"olicher Lie group, then $G$ is a Fr\"olicher Lie group.
\end{Theorem}

\noindent 
This result was proven in \cite{Ma2016} by applying Theorem 
\ref{exactsequence}. Using the equivalence between 
Gateaux-smooth and Fr\"olicher-smooth in the Fr\'echet category 
stated after Definition 2 and proven in \cite{MR2016}, 
we have a Fr\'echet version of Theorem \ref{DiffPDO}: if $H$ is 
a regular Fr\'echet Lie group which is stable under 
$Diff(S^1)-$adjoint action, and $G$ is a smooth Fr\'echet manifold 
isomorphic to $H \times Diff(S^1)$ with multiplication and 
inversion Fr\"olicher (hence Fr\'echet) smooth, we have the 
equivalence:
$$H 
\hbox{ is a {\em regular} Fr\'echet Lie group } \Leftrightarrow G 
\hbox{ is a {\em regular} Fr\'echet Lie group\; .}$$  

The pseudodifferential operators considered in Theorem 
\ref{DiffPDO} can be classical, odd class, or anything else. 
Applying the formulas of ``changes of coordinates'' (which can be 
understood as adjoint actions of diffeomorphisms) of e.g. 
\cite{Gil}, we obtain that odd-class pseudodifferential operators 
are stable under the adjoint action of $Diff(S^1).$  Thus, we can 
define the following group:

\begin{Definition}
	The group $FCl_{Diff(S^1),odd}^{0,*}(S^1,V)$ is the regular 
	Fr\'echet Lie group $G$ obtained in Theorem 
	$\ref{DiffPDO}$ with $H=Cl^{0,*}_{odd}(S^1,V).$  
\end{Definition}

Following \cite{Ma2016}, we remark that operators $A$ in this group 
can be understood as operators in $Cl^{0,*}_{odd}(S^1,V)$ twisted 
by diffeomorphisms, this is, 
\begin{equation} \label{aux}
A = B \circ g 
\end{equation}
for unique $g \in Diff(S^1)$ and unique 
$B \in Cl^{0,*}_{odd}(S^1,V)$, 
and also that its Lie algebra is isomorphic as a vector space to
$Cl^0_{odd}(S^1,V)\oplus Vect(S^1)$, in which $Vect(S^1)$ is the 
space of smooth vector fields on $S^1$.

\begin{remark} \label{Baaj} 
	The diffeomorphism $g$ appearing in $(\ref{aux})$ is
	the phase of the operators, but here the phase (and hence the 
	decomposition $(\ref{aux})$) is unique, which 
	is not the case for general Fourier integral operators, see e.g. 
	\cite{Horm}.
	This construction of phase functions of 
	$Diff(M)-$pseudodifferential operators
	differs from the one described by Omori \cite{Om} 
	and Adams, Ratiu and Schmid \cite{ARS2} for the groups of Fourier 
	integral operators; the exact relation among these
	constructions still needs to be investigated.
\end{remark} 

Now we note that 
the group $Diff(S^1)$ decomposes into two connected components 
$Diff(S^1) = Diff_+(S^1) \cup Diff_-(S^1),$
where the connected component of the identity, $Diff_+(S^1)$, is 
the group of orientation preserving diffeomorphisms 
of $S^1$. We make the following definition:

\begin{Definition}
	The group $FCl_{Diff_+(S^1),odd}^{0,*}(S^1,V)$ is the regular 
	Fr\'echet Lie group of all operators in 
	$FCl_{Diff(S^1),odd}^{0,*}(S^1,V)$ whose phase diffeomorphisms
	lie in the group $Diff_+(S^1).$
\end{Definition}

\begin{Theorem}\label{SY}
	Let $U \in FCl^{0,*}_{Diff_+(S^1),odd}(S^1,V).$ There exists an 
	unique pair 
	$$(S,Y) \in Cl^{-1,*}_{odd}(S^1,V)\times 
	\left(DO^{0,*}(S^1,V) \rtimes Diff_+(S^1)\right)$$
	such that $$U = S\,Y\;.$$
	Moreover, the map $U \mapsto (S,Y)$ is smooth and, there is a short 
	exact sequence of Lie groups:
	$$
	{ 1} \rightarrow Cl^{-1,*}_{odd}(S^1,V) \rightarrow 
	FCl_{Diff_+(S^1),odd}^{0,*}(S^1,V) \rightarrow 
	DO^0(S^1,V)\rtimes Diff_+(S^1) \rightarrow { 1}$$
	for which the $Y$-part defines a smooth global section, and which 
	is a morphism of groups.
\end{Theorem}
\begin{proof}
	We already know that $U$ splits in an unique way as
	$ U = A_0\,.\,g\; ,$
	in which $g \in Diff_+(S^1)$ and $A_0 \in Cl^{0,*}_{odd}(S^1,V).$
	By Proposition \ref{SD}, the pseudodifferential operator $A_0$ can 
	be written uniquely as a sum, $A = A_I + A_D$, 
	in which $A_D \in DO^0(S^1,V) \subset Cl_{odd}(S^1,V)$.
	Since $A_0$ is invertible, $\sigma_0(A_0) \in C^\infty(S^1,GL(V))$ 
	and hence $A_D \in DO^{0,*}(S^1,V).$ We can write  
	$$U = A_0. A_D^{-1}.A_D.g.$$
	We get $Y=A_D.g \in DO^{0,*}(S^1,V)\rtimes Diff_{+}(S^1)$ and 
	$S = A_0. A_D^{-1} \in Cl^{0,*}_{odd}(S^1,V)$    
	(the inverse of an odd class operator is an odd class operator).  
	Let us compute the principal symbol $\sigma_0(S)$:
	$$\sigma_0(S)= \sigma_0(A_0)\sigma_0(A_D^{-1})= 
	\sigma_0(A_0)\sigma_0(A_0)^{-1} = Id_V\; .$$
	Thus, $S \in Cl^{-1,*}_{odd}(S^1,V).$ Moreover, the maps 
	$U \mapsto g$ and  $A_0 \mapsto A_D$ are smooth, and this 
	observation ends the proof.
\end{proof}

Let us summarize our constructions. The semi-direct product of 
{Fr\'echet} Lie groups 
$$
FCl^{0,*}_{Diff_+(S^1),odd}(S^1,V) = 
Cl^{0,*}_{odd}(S^1,V)\rtimes Diff_+(S^1)
$$
\smallskip
fully described by the exact sequence
$$\begin{array}{ccccccccc}
1 & \rightarrow & Cl^{0,*}_{odd}(S^1,V) & \rightarrow & 
FCl^{0,*}_{Diff_+(S^1),odd}(S^1,V) & \rightarrow & Diff_+(S^1) & 
\rightarrow & 1
\end{array}
$$
 and by the associated sequence of Lie 
algebras
$$\begin{array}{ccccccccc}
0 & \rightarrow & Cl^{0}_{odd}(S^1,V) & \rightarrow & Cl^{0}(S^1,V) \rtimes Vect(S^1) & \rightarrow & Vect(S^1) & \rightarrow & 0\; ,
\end{array}
$$
{
in which we have used (\ref{alfa}) and (\ref{beta0}) in order to understand differential operators having symbols of order 1 as elements of $Vect(S^1)\otimes Id_V$,  
}
can be completed by the following diagram in which vertical and 
horizontal lines are short exact sequences of Lie groups:
$$\begin{array}{ccccccccc}
&&&&1&&1&&\\
&&&&\downarrow&&\downarrow&&\\
1 & \rightarrow & Cl^{-1,*}_{odd}(S^1,V) & \rightarrow & Cl^{0,*}_{odd}(S^1,V) & \rightarrow & DO^{0,\ast}(S^1,V)  & \rightarrow & 1\\
&&\|&&\downarrow &&\downarrow &&\\
1 & \rightarrow & Cl^{-1,*}_{odd}(S^1,V) & \rightarrow & FCl^{0,*}_{Diff_+(S^1),odd}(S^1,V) & \rightarrow & DO^{0,\ast}(S^1,V)\rtimes Diff_+(S^1) & \rightarrow & 1\\
&&&&\downarrow&&\downarrow&&\\
&&&&Diff_+(S^1)&=&Diff_+(S^1)&&\\
&&&&\downarrow&&\downarrow&&\\
&&&&1&&1&&\\
\end{array}
$$

The corresponding diagram of Lie algebras, all of them embedded in 
$Cl_{odd}(S^1,V)$ is:

$$\begin{array}{ccccccccc}
&&&&0&&0&&\\
&&&&\downarrow&&\downarrow&&\\
0 & \rightarrow & Cl^{-1}_{odd}(S^1,V) & \rightarrow & Cl^{0}_{odd}(S^1,V) & \rightarrow & DO^{0}(S^1,V)  & \rightarrow & 0\\
&&\|&&\downarrow &&\downarrow &&\\
0 & \rightarrow & Cl^{-1}_{odd}(S^1,V) & \rightarrow & Cl^{0}_{odd}(S^1,V) \rtimes Vect(S^1) & \rightarrow & DO^{0}(S^1,V)\rtimes Vect(S^1) & \rightarrow & 0\\
&&&&\downarrow&&\downarrow&&\\
&&&&Vect(S^1)&=&Vect(S^1)&&\\
&&&&\downarrow&&\downarrow&&\\
&&&&0&&0&&\\
\end{array}
$$

We end this appendix by considering exponential mappings. We can do so, 
since the Lie groups $Cl^{-1,*}_{odd}(S^1,V),$  
$FCl_{Diff_+(S^1),odd}^{0,*}(S^1,V)$ and 
$DO^0(S^1,V)\rtimes Diff_+(S^1)$ are regular (see our discussion at 
the beginning of this section and Definition 24). 
Let us consider a curve $L(t)$ in the Lie algebra of 
$FCl_{Diff_+(S^1),odd}^{0,*}(S^1,V)$ which, thanks to Proposition 
25 and the Mulase decomposition, we can identify (as a vector 
space) with $Cl^{-1}_{odd}(S^1,V)\oplus DO^1(S^1,V)$.
Thus, we assume 
$$ 
L(t) \in C^\infty([0;1], Cl^{-1}_{odd}(S^1,V)\oplus DO^1(S^1,V) )  
$$ 
and we write $L(t) = L_D(t) + L_S(t)$. 
We compare the exponential 
$\exp(L)(t) \in C^\infty([0;1],FCl^{0,*}_{Diff_+(S^1),odd}(S^1,V))$ 
with 
$$
\exp(L_D)(t) \in C^\infty \left( [0;1],DO^{0,*}(S^1,V))\rtimes 
Diff_+(S^1) \right)
$$
and $$
\exp(L_S)(t) \in C^\infty([0;1],Cl^{-1,*}_{odd}(S^1,V))\; .
$$
On the one hand, we can write 
$$
\exp(L)(t) = S(t)Y(t)
$$ 
according to Theorem \ref{SY}, and we know that the paths 
$t\mapsto S(t)$ 
and $t \mapsto Y(t)$ are smooth. On the other hand, using the 
definition of the left exponential map, we get
$$
\frac{d}{dt}\exp(L)(t) = exp(L)(t). L(t) \; .
$$
Thus, gathering the last two expressions we obtain 
\begin{eqnarray*}
	\frac{d}{dt}\exp(L)(t) & = & \frac{d}{dt}\left( S(t) Y(t)\right) \\
	& = & \left(\frac{d}{dt}S(t)\right)S^{-1}(t)S(t) Y(t)+ S(t)Y(t)Y^{-1}(t)\left(\frac{d}{dt}Y(t)\right)\\
	& = & \left(\frac{d}{dt}S(t)S^{-1}(t)\right) \exp(L)(t)+ \exp(L)(t)Y^{-1}(t)\left(\frac{d}{dt}Y(t)\right) \\
	& = & \exp(L)(t)\left(Ad_{\exp(L)(t)^{-1}}\left( \left(\frac{d}{dt}S(t)S^{-1}(t)\right)\right) + 
	Y^{-1}(t)\left(\frac{d}{dt}Y(t)\right)  \right)\; .
\end{eqnarray*}
Now, $Y^{-1}(t) \frac{d}{dt}Y(t)$ is a smooth path on the space of 
differential operators of order 1, and we have 
$$
Ad_{\exp(L)(t)^{-1}}\left( \left(\frac{d}{dt}S(t)S^{-1}(t)\right)
\right) \in Cl^{-1}_{odd}(S^1,V) \; .
$$
These calculations allow us to prove the following:

\begin{Proposition}
	Let us assume that $L(t)$ is a curve in the Lie algebra of the 
	group $FCl^{0,*}_{Diff_+(S^1),odd}(S^1,V)$, that
	$L(t) = L_S(t) + L_D(t)$ with $L_S(t) \in Cl^{-1}_{odd}(S^1,V)$ 
	and $L_D(t) \in DO^1(S^1,V)$, and that $\exp(L)(t) = S(t) Y(t)$. 
	Then, 
	$$Y(t) = \exp(L_D)(t)$$
	and 
	$$ S(t) = \exp\left(Ad_{\exp(L)(t)}\left(L_S\right)\right)(t)\; .$$
\end{Proposition} 
\begin{proof}
	We have already obtained that
	$$L_D = Y(t)^{-1}\frac{d}{dt}Y(t)$$
	and that 
	$$L_S =
	Ad_{\exp(L)(t)^{-1}}\left( \left(\frac{d}{dt}S(t)S^{-1}(t)\right)
	\right)
	$$
	because of the uniqueness of the decomposition $$L=L_S + L_D\; .$$
	We obtain the result by  passing to the exponential maps on the 
	groups $Cl^{-1,*}_{odd}(S^1,V)$ and 
	$DO^{0,*}(S^1,V)\rtimes Diff_+(S^1).$ 
\end{proof}}
\vskip 12pt	

\paragraph{\bf Acknowledgements:} Both authors have been partially
supported by CONICYT (Chile) via the {\em Fondo Nacional de 
Desarrollo Cient\'{\i}fico y Tecnol\'{o}gico} operating grants 
\# 1161691 and \# 1201894. The authors would like to thank Saad 
Baaj for comments leading to Remark \ref{Baaj}.

\end{document}